\documentclass{article}%
\usepackage{amsmath}
\usepackage{amssymb}
\usepackage{amsfonts}
\usepackage{amsthm}
\usepackage{nccmath}
\usepackage{fullpage}
\usepackage{array,longtable}
\usepackage{graphicx}
\usepackage{color}
\usepackage{multicol}
\usepackage{multirow}%
\providecommand{\U}[1]{\protect\rule{.1in}{.1in}}
\newtheorem{theorem}{Theorem}
\newtheorem{lemma}[theorem]{Lemma}

\newtheorem{definition}[theorem]{Definition}
\newtheorem{assumption}{Assumption}
\newtheorem{problem}{Problem}

\newtheorem{example}[theorem]{Example}
\newtheorem{remark}[theorem]{Remark}

\allowdisplaybreaks
\def\revise#1{{#1}}
\begin{document}

\title{A Condition for Multiplicity Structure of Univariate Polynomials}
\author{Hoon Hong\\Department of Mathematics, North Carolina State University\\Box 8205, Raleigh, NC 27695, USA\\hong@ncsu.edu\\[10pt] Jing Yang\thanks{Corresponding author.}\\SMS--KLSE--School of Mathemetics and Physics, Guangxi University for Nationalities\\Nanning 530006, China\\yangjing0930@gmail.com\\[-10pt] }
\date{}
\maketitle

\begin{abstract}
We consider the problem of finding a condition for a univariate polynomial
having a given multiplicity structure when the number of distinct roots is
given. It is well known that such conditions can be written as conjunctions of
several polynomial equations and one inequation in the coefficients, by using
repeated parametric gcd's. In this paper, we give a novel condition which is
not based on repeated gcd's. Furthermore, it is shown that the number of
polynomials in the condition is optimal and the degree of polynomials is
smaller than that in the previous condition based on repeated gcd's.

\end{abstract}

\section{Introduction}

In this paper, we consider the problem of finding a condition on the
coefficients of a polynomial \revise{over the complex field $\mathbb{C}$} so that it has a given multiplicity structure.
For example, consider a quartic polynomial $F=a_{4}x^{4}+a_{3}x^{3}+a_{2}%
x^{2}+a_{1}x+a_{0}$ \revise{where $a_i$'s take values over $\mathbb{C}$}. We would like to find a condition on $a_{i}$'s so that
$F$ has the multiplicity structure $(3,1)$, that is, it has two distinct
complex roots, say $r_{1}$ and $r_{2}$, where the multiplicities of $r_{1}$
and $r_{2}$ are $3$ and $1$ respectively. The problem is important because
many tasks in mathematics, science and engineering can be reduced to the
problem. A prerequisite for the problem is finding a condition on coefficients
such that the polynomials has the given number of distinct roots. This is already well studied. For instance, the subdiscriminant theory provides a
complete solution to the sub-problem. More explicitly, a univariate polynomial
of degree $n$ has $m$ distinct roots if and only if its $0$-th,$\ldots
$,$(n-m-1)$-th psd's (i.e., principal subdiscriminant coefficient) vanish and
the $(n-m)$-th psd does not. For details, see standard textbooks on
computational algebra (e.g., \cite{2006_Basu_Pollack_Roy}).

Thus from now on, we will assume that the number of distinct roots is fixed,
say $m$. However, even with this assumption, there can be several different
multiplicity structures. For example, consider again a quartic univariate
polynomial $F$. Assume that it has two distinct roots. Then its multiplicity
may be $(3,1)$ or $(2,2)$. This naturally leads to the problem: how to
discriminate the two cases? In general, the problem is stated as follows:

\medskip

\textbf{Problem}: \emph{Let }$\mu=\left(  \mu_{1},\ldots,\mu_{m}\right)
$\emph{ }be such that $\mu_{1},\ldots,\mu_{m}\geq1\ $and $\mu_{1}+\cdots
+\mu_{m}=n$\emph{. Find a condition on the coefficients of a polynomial }%
$F$\emph{  \revise{over  $\mathbb{C}$} of degree~}$n$\emph{ with }$m$\emph{ distinct \revise{complex} roots so that the
multiplicity structure of }$F$\emph{ is }$\mu$\emph{. (We will call the
condition a }$\mu$-\emph{multiplicity-discriminating condition.)}

\medskip

Due to its importance, the problem and several related problems have been
already carefully studied. In \cite{1996_Yang_Hou_Zeng}, Yang, Hou and Zeng
gave an algorithm to generate a multiplicity-discriminating condition
(referred as YHZ's condition hereinafter) by making use of repeated gcd
computation for parametric polynomials
\cite{1971_Brown,1967_Collins,1982_Loos}. It is based on a similar idea
adopted by Gonzalez-Vega et al. \cite{1998_Gonzalez_Recio_Lombardi} for
solving the real root classification and quantifier elimination problems by
using Sturm-Habicht sequences. YHZ's work was followed by Liang and Zhang
\cite{1999_Liang_Zhang} who solved the root classification of polynomials with
the form $F(x)+I\cdot G(x)$ where $F,G\in\mathbb{R}[x]$ and $I$ is the
imaginary unit. Further improvement and generalization can be found in
\cite{2006_Liang_Jefferey,2008_Liang_Jefferey_Maza}. \revise{Multiplicity structure is a particular root configuration of a univariate polynomial. In \cite{2019_Wang_Yang}, another particular root configuration is studied where there exists a symmetric triple of roots among which one root is the average of the other two. }

It is known that a multiplicity-discriminating condition can be written as a
conjunction of several polynomial equations and one inequation on the
coefficients. For example, for a quartic polynomial with two distinct roots,
see two different conditions in Example \ref{ex:running}. In general, there
are infinitely many syntactically different conditions. Thus a challenge is to
find a condition with \textquotedblleft small\textquotedblright\ size. A
natural way to measure the \textquotedblleft size\textquotedblright of the
condition is the number of polynomials appearing in the condition and their
maximum degree.

The main contribution in this paper is to provide a condition with \emph{only
one} polynomial with degree \emph{smaller} than those in the previous method.
The condition is \emph{novel} in that it is based on a significantly different
theory and techniques from the previous methods (which are essentially based
on repeated parametric gcd or subdiscriminant theory). In order to find the
new condition we developed the following ideas and techniques.

\begin{enumerate}
\item Convert the multiplicity condition in roots into an equivalent
permanental equation in roots.

\item Convert the permanent in roots into a sum of determinants in roots.

\item Convert each determinant in roots into a determinant in coefficients.
\end{enumerate}

\noindent We found that the above ideas/techniques are interesting on their
own. We hope that they could be useful for tackling other related problems.

The paper is structured as follows. In Section \ref{sec:main}, we give a
precise statement of the main result of the paper (Theorem
\ref{thm:multdiscriminant}). In Section \ref{sec:proof}, we give a proof of
Theorem \ref{thm:multdiscriminant}. The proof is long thus we divide the proof
into three subsections which are interesting on their own. In Section
\ref{sec:compare}, we compare the sizes of the multiplicity-discriminant
condition in Theorem \ref{thm:multdiscriminant} and that given by a previous work.

\section{Main Results}

\label{sec:main} In this section, we give a precise statement of the main
result of the paper. For this, we need a few notions and notations.

\begin{definition}
[Multiplicity of a polynomial]
\revise{Let $\mathbb{C}$ be the complex field }and $F$ $\in\mathbb{C}\left[  x\right]  $ be
with $m$ distinct complex roots, say $r_{1},\ldots,r_{m}$. The multiplicity of
$F$, written as $\operatorname*{mult}\left(  F\right)  $, is defined by%
\[
\operatorname*{mult}\left(  F\right)  =\left(  \mu_{1},\ldots,\mu_{m}\right)
\]
where $\mu_{i}$ is the multiplicity of $r_{i}$ as a root of $F$. Without
losing generality, we assume that $\mu_{1}\geq\cdots\geq\mu_{m}$.
\end{definition}

\begin{assumption}
We assume that $2\leq m\leq n-2$.
\end{assumption}

\begin{remark}
The assumption is natural and meaningful because otherwise there is nothing to
discriminate: If $m=1$ then the only possible multiplicity is $(n)$. If $m=n$
then the only possible multiplicity is $(1,\ldots,1)$. If $m=n-1$ then the
only possible multiplicity is $(2,1,\ldots,1)$.
\end{remark}

\begin{problem}
\revise{Let $n\geq m$ be fixed.}

\begin{description}
\item[Input:]
\revise{$\mu=\left(  \mu_{1},\ldots,\mu_{m}\right)  \ \ $such that $\mu
_{1}\geq\cdots\geq\mu_{m}\geq1$ and $\mu_{1}+\cdots+\mu_{m}=n$ }

\item[Output:]
\revise{\emph{a }$\mu$-\emph{multiplicity-discriminating condition, that
is, a condition on the coefficients of a polynomial }$F$\emph{ of degree~}$n$\emph{ with }$m$\emph{ distinct complex roots so that the multiplicity structure of
}$F$\emph{ is }$\mu$\emph{. }}
\end{description}
\end{problem}

\begin{example}
\label{ex:running} Let $F(x)=a_{4}x^{4}+a_{3}x^{3}+a_{2}x^{2}+a_{1}x+a_{0}$ be
such that
\revise{$\deg F=4$ and the number of distinct roots of $F$ is $2$}.
The followings two are $\left(  3,1\right)  $-multiplicity discriminating
conditions on $F$.

\begin{enumerate}
\item $C_{1}=0\wedge C_{2}\neq0$ where

\quad$C_{1}=-36\,{a_{4}^{3}}{a_{{3}}^{5}}a_{{1}}+12\,{a_{4}^{3}}{a_{{3}}^{4}%
}{a_{{2}}^{2}}+576\,{a_{4}^{4}}{a_{{3}}^{4}}a_{{0}}+48\,{a_{4}^{4}}{a_{{3}%
}^{3}}a_{{2}}a_{{1}}-32\,{a_{4}^{4}}{a_{{3}}^{2}}{a_{{2}}^{3}}-3072\,{a_{4}%
^{5}}{a_{{3}}^{2}}a_{{2}}a_{{0}}$

\quad\quad\quad\ $+432\,{a_{4}^{5}}{a_{{3}}^{2}}{a_{{1}}^{2}}+128\,{a_{4}^{5}%
}a_{{3}}{a_{{2}}^{2}}a_{{1}}+4096\,{a_{4}^{6}}{a_{{2}}^{2}}a_{{0}%
}-1152\,{a_{4}^{6}}a_{{2}}{a_{{1}}^{2}}$

\quad$C_{2}=16\,{a_{4}^{2}}a_{{2}}-6\,{a_{4}}{a_{{3}}^{2}}.$

\item $C_{1}^{\prime}\neq0$ where

\quad$C_{1}^{\prime}=-64\,{a_{{4}}^{5}}{a_{{1}}^{2}}+64\,{a_{{4}}^{4}}a_{{3}%
}a_{{2}}a_{{1}}-16\,{a_{{4}}^{3}}{a_{{3}}^{2}}{a_{{2}}^{2}}+8\,{a_{{4}}%
^{2}{a_{{3}}^{4}}a_{{2}}}-16\,{a_{{4}}^{2}}{a_{{3}}^{3}}a_{{1}}-{a_{{4}}%
}{a_{{3}}^{6}}$.
\end{enumerate}
\end{example}

\begin{remark}
As you see in the above example, in general, the $\mu$-multiplicity
discriminating condition of $F$ is not unique syntactically. In fact, there
are infinitely many syntactically different $\mu$-multiplicity discriminant
conditions. Thus a challenge is to find a syntactically \textquotedblleft
small\textquotedblright\ condition.
\end{remark}

\begin{definition}
[Determinant of polynomials]Let $F_{0},\ldots,F_{k}\in\mathbb{C}\left[
x\right]  $ be such that $\deg F_{0},\ldots,\deg F_{k}\leq k$. Then the
\emph{determinant} \emph{of the polynomials} $\operatorname*{dp}$ is defined
by
\[
\operatorname*{dp}\left[
\begin{array}
[c]{c}%
F_{0}\\
\vdots\\
F_{k}%
\end{array}
\right]  =\det\left[
\begin{array}
[c]{ccc}%
a_{0,k} & \cdots & a_{0,0}\\
\vdots &  & \vdots\\
a_{k,k} & \cdots & a_{k,0}%
\end{array}
\right]
\]
where $F_{i}=\sum\limits_{0\leq j\leq k}a_{i,j}x^{j}$.
\end{definition}

\noindent We have introduced all the necessary notions and notations, and
thus, now we give a precise statement of the main result of this paper.

\begin{theorem}
[Main result]\label{thm:multdiscriminant}
\revise{
Let $\mu=\left(  \mu
_{1},\ldots,\mu_{m}\right)$ be such that $\mu_{1},\ldots,\mu_{m}\geq1$ 
and $\mu_{1}+\cdots+\mu_{m}=n$,
and $F\in\mathbb{C}\left[ x\right]$ be of degree $n$. }
Let
\[
D_{\mu}\left(  F\right)  =\sum_{{\sigma}\in S_{p}}\operatorname*{dp}\left[
\begin{array}
[c]{c}%
x^{n-\mu_{m}-1}F\\
\vdots\\
x^{0}F\\
x^{n-1}F^{\left(  \sigma_{_{1}}\right)  }/\sigma_{1}!\\
\vdots\\
x^{0}F^{\left(  \sigma_{_{n}}\right)  }/\sigma_{n}!
\end{array}
\right]
\]
where $p=(\underset{\mu_{1}}{\underbrace{\mu_{1},\ldots,\mu_{1}}},\ldots,$
$\underset{\mu_{m}}{\underbrace{\mu_{m},\ldots,\mu_{m}}})$, $S_{p}$ is the set
of all permutations of $p$
\revise{and $F^{(i)}$ is the $i$-th derivative of $F$ in terms of $x$}.
Then $D_{\mu}(F)\neq0$ is a $\mu$-multiplicity-discriminanting condition.
\end{theorem}

\begin{remark}
\label{rem:degree_of_D}
\revise{Assume $F=\sum_{i=0}^na_ix^i$. Then a straightforward
degree analysis of the expression of $D_{\mu}$ in Theorem
\ref{thm:multdiscriminant} shows that the degree of $D_{\mu}(F)$ in $a$ is
$2\,n-\mu_{m}$ where $a=(a_0,a_1,\ldots,a_n)$.}
\end{remark}

\begin{example}
\label{ex:quartic-poly} Let $F(x)=a_{4}x^{4}+a_{3}x^{3}+a_{2}x^{2}%
+a_{1}x+a_{0}$ be such that
\revise{$\deg F=4$ and the number of distinct roots of $F$ is
$2$}. We will construct a $\left(  3,1\right)  $-multiplicity discriminating
condition on $F$, using the main result (Theorem \ref{thm:multdiscriminant}). Note

\begin{enumerate}
\item $\mu=(3,1)$

\item $p=(3,3,3,1)$

\item $S_{p}=\{(3,3,3,1),(3,3,1,3),(3,1,3,3),(1,3,3,3)\}$

\item
\begin{align*}
D_{(3,1)}(F)=  &  \,\operatorname*{dp}\left[
\begin{array}
[c]{c}%
x^{2}F\\[5pt]%
x^{1}F\\[5pt]%
x^{0}F\\[5pt]%
x^{3}F^{(3)}/3!\\[5pt]%
x^{2}F^{(3)}/3!\\[5pt]%
x^{1}F^{(3)}/3!\\[5pt]%
x^{0}F^{(1)}/1!
\end{array}
\right]  +\operatorname*{dp}\left[
\begin{array}
[c]{c}%
x^{2}F\\[5pt]%
x^{1}F\\[5pt]%
x^{0}F\\[5pt]%
x^{3}F^{(3)}/3!\\[5pt]%
x^{2}F^{(3)}/3!\\[5pt]%
x^{1}F^{(1)}/1!\\[5pt]%
x^{0}F^{(3)}/3!
\end{array}
\right]  +\operatorname*{dp}\left[
\begin{array}
[c]{c}%
x^{2}F\\[5pt]%
x^{1}F\\[5pt]%
x^{0}F\\[5pt]%
x^{3}F^{(3)}/3!\\[5pt]%
x^{2}F^{(1)}/1!\\[5pt]%
x^{1}F^{(3)}/3!\\[5pt]%
x^{0}F^{(3)}/3!
\end{array}
\right]  +\operatorname*{dp}\left[
\begin{array}
[c]{c}%
x^{2}F\\[5pt]%
x^{1}F\\[5pt]%
x^{0}F\\[5pt]%
x^{3}F^{(1)}/1!\\[5pt]%
x^{2}F^{(3)}/3!\\[5pt]%
x^{1}F^{(3)}/3!\\[5pt]%
x^{0}F^{(3)}/3!
\end{array}
\right] \\[5pt]
=  &  \,\det\left[
\begin{array}
[c]{ccccccc}%
a_{4} & a_{3} & a_{2} & a_{1} & a_{0} &  & \\
& a_{4} & a_{3} & a_{2} & a_{1} & a_{0} & \\
&  & a_{4} & a_{3} & a_{2} & a_{1} & a_{0}\\
&  & 4a_{4} & a_{3} &  &  & \\
&  &  & 4a_{4} & a_{3} &  & \\
&  &  &  & 4a_{4} & a_{3} & \\
&  &  & 4a_{4} & 3\,a_{3} & 2\,a_{2} & a_{1}%
\end{array}
\right]  +\det\left[
\begin{array}
[c]{ccccccc}%
a_{4} & a_{3} & a_{2} & a_{1} & a_{0} &  & \\
& a_{4} & a_{3} & a_{2} & a_{1} & a_{0} & \\
&  & a_{4} & a_{3} & a_{2} & a_{1} & a_{0}\\
&  & 4a_{4} & a_{3} &  &  & \\
&  &  & 4a_{4} & a_{3} &  & \\
&  & 4a_{4} & 3\,a_{3} & 2\,a_{2} & a_{1} & \\
&  &  &  &  & 4a_{4} & a_{3}%
\end{array}
\right]  +\\
&  \,\det\left[
\begin{array}
[c]{ccccccc}%
a_{4} & a_{3} & a_{2} & a_{1} & a_{0} &  & \\
& a_{4} & a_{3} & a_{2} & a_{1} & a_{0} & \\
&  & a_{4} & a_{3} & a_{2} & a_{1} & a_{0}\\
&  & 4a_{4} & a_{3} &  &  & \\
& 4a_{4} & 3\,a_{3} & 2\,a_{2} & a_{1} &  & \\
&  &  &  & 4a_{4} & a_{3} & \\
&  &  &  &  & 4a_{4} & a_{3}%
\end{array}
\right]  +\det\left[
\begin{array}
[c]{ccccccc}%
a_{4} & a_{3} & a_{2} & a_{1} & a_{0} &  & \\
& a_{4} & a_{3} & a_{2} & a_{1} & a_{0} & \\
&  & a_{4} & a_{3} & a_{2} & a_{1} & a_{0}\\
4a_{4} & 3\,a_{3} & 2\,a_{2} & a_{1} &  &  & \\
&  &  & 4a_{4} & a_{3} &  & \\
&  &  &  & 4a_{4} & a_{3} & \\
&  &  &  &  & 4a_{4} & a_{3}%
\end{array}
\right] \\
=  &  \,-64\,{a_{{4}}^{5}}{a_{{1}}^{2}}+64\,{a_{{4}}^{4}}a_{{3}}a_{{2}}a_{{1}%
}-16\,{a_{{4}}^{3}}{a_{{3}}^{2}}{a_{{2}}^{2}}+8\,{a_{{4}}^{2}{a_{{3}}^{4}%
}a_{{2}}}-16\,{a_{{4}}^{2}}{a_{{3}}^{3}}a_{{1}}-{a_{{4}}}{a_{{3}}^{6}}%
\end{align*}

Note that it is the polynomial $C_{1}^{\prime}$ in Example \ref{ex:running}.

\item The main result (Theorem \ref{thm:multdiscriminant}) states that
$D_{(3,1)}(F)\neq0$ is a $\left(  3,1\right)  $-multiplicity discriminating
condition on $F$.
\end{enumerate}
\end{example}

\begin{remark}
Observe $D_{(1,\ldots,1)}(F)$ is the Sylvester resultant of $F$ and
$F^{\prime}$. Thus $D_{\mu}$ can be viewed as a certain generalization of
Sylvester resultant of $F$ and $F^{\prime}$ \revise{(i.e., the traditional discriminant of $F$ up to sign)}.
\end{remark}

\section{Proof of Main Result (Theorem \ref{thm:multdiscriminant})}

\label{sec:proof} Let $F\in\mathbb{C}[x]$ be of degree $n$ and $\alpha
_{1},\ldots,\alpha_{n}$ be the $n$ roots of $F$. In this section, we give a
proof of Theorem \ref{thm:multdiscriminant}. The proof is long thus we divide
the proof into three steps (lemmas), which are interesting on their own.

\begin{enumerate}
\item Lemma \ref{lem:expr-in-roots}: We show that the multiplicity condition
in roots can be converted into an equivalent polynomial inequation which is a
permanental expression in roots.

\item Lemma \ref{lem:gen-D-ratio}: We show that the permanent in roots can be
converted into a sum of determinants in roots.

\item Lemma \ref{lem:ratio}: We show that each determinant in roots can be
converted into a determinant in coefficients.
\end{enumerate}

\noindent Finally, we combine the above three lemmas to prove the main result.

\subsection{From a condition in roots to a permanental condition in roots}

\begin{lemma}
\label{lem:expr-in-roots}
\revise{Let $F$ be of degree $n$ and $\alpha_1,\ldots,\alpha_n$ be its $n$ complex roots. Then}
\[
\operatorname*{mult}\left(  F\right)  =\mu\ \ \Longleftrightarrow
\ \ \overline{D}_{\mu}\left(  F\right)  \neq0
\]
where

\begin{itemize}
\item $\overline{D}_{\mu}\left(  F\right)  =\operatorname*{per}\left[
\begin{array}
[c]{ccc}%
\dfrac{F^{(p_{1})}(\alpha_{1})}{p_{1}!} & \cdots & \dfrac{F^{(p_{1})}%
(\alpha_{n})}{p_{1}!}\\[8pt]%
\vdots &  & \vdots\\[8pt]%
\dfrac{F^{(p_{n})}(\alpha_{1})}{p_{n}!} & \cdots & \dfrac{F^{(p_{n})}%
(\alpha_{n})}{p_{n}!}%
\end{array}
\right]  \ /\ c; $

\item \revise{$\operatorname*{per}$ stands for the permanent operation;}

\item $p=(\underset{\mu_{1}}{\underbrace{\mu_{1},\ldots,\mu_{1}}},\ldots,$
$\underset{\mu_{m}}{\underbrace{\mu_{m},\ldots,\mu_{m}}})$;

\item $c$ is an integer determined by $p$; specifically if $p$ consists of
$\ell$ distinct numbers occurring $q_{1},\ldots,q_{\ell}$ times, then
$c=\prod_{i=1}^{\ell}q_{i}!$.
\end{itemize}
\end{lemma}

\begin{proof}
\revise{We will prove the lemma by ``deriving'' the
condition$\ \overline{D}_{\mu}\left(  F\right)  \neq0$ instead of merely
``verifying'' the correctness of the lemma, since it will be much more interesting
to read, bringing out the underlying ideas and intuitions. The derivation will
be driven by two wishes:
\begin{enumerate}
\item Wish to find a condition that involves a \emph{single} polynomial on the all the (not necessarily distinct) roots $\alpha_1,\ldots,\alpha_n$.
\item Wish to find the polynomial which is \emph{symmetric} in the roots,
so that later we can turn it into an expression in the coefficients.
\end{enumerate}
\noindent
The strategy is to repeatedly rewrite the condition $\operatorname*{mult}%
\left(  F\right)  =\mu$ with the above two wishes in mind.
}
\begin{enumerate}
\item \revise{Note that the condition $\operatorname*{mult}\left(  F\right)  =\mu\ $is written in
terms of \emph{distinct} roots. We rewrite the condition into a {\em symmetric\/} condition on   {\em all\/} the (not necessarily distinct)\ roots  $\alpha_1,\ldots,\alpha_n$.
\[
\operatorname*{mult}\left(  F\right)  =\mu\ \ \Longleftrightarrow
\ \ \bigvee\limits_{\sigma\in S_{p}}\bigwedge\limits_{i=1}^{n}%
\operatorname*{mult}\left(  \alpha_{i}\right)  =\sigma_{i}%
\]
Proof: Obvious from the definition of multiplicity of a root.
}

\item \revise{We rewrite the symmetric condition into a symmetric {\em polynomial\/} condition
\[
\bigvee\limits_{\sigma\in S_{p}}\ \bigwedge\limits_{i=1}^{n}%
\operatorname*{mult}\left(  \alpha_{i}\right)  =\sigma_{i}%
\ \ \ \ \Longleftrightarrow\ \ \bigvee\limits_{\sigma\in S_{p}}\ \ \bigwedge
\limits_{i=1}^{n}F^{\left(  \sigma_{i}\right)  }\left(  \alpha_{i}\right)
\neq0
\]
}

\revise{Proof: We will show each direction of $\Longleftrightarrow$ one by one.}

\begin{enumerate}
\item \revise{$\bigvee\limits_{\sigma\in S_{p}}\bigwedge\limits_{i=1}^{n}%
\operatorname*{mult}\left(  \alpha_{i}\right)  =\sigma_{i}\ \Longrightarrow
\ \ \bigvee\limits_{\sigma\in S_{p}}\bigwedge\limits_{i=1}^{n}F^{\left(
\sigma_{i}\right)  }\left(  \alpha_{i}\right)  \neq0$}

\revise{From the elementary  calculus, we have
\[
\operatorname*{mult}\left(  \alpha_{i}\right)  =\sigma_{i}\ \Longrightarrow
\ F^{\left(  \sigma_{i}\right)  }\left(  \alpha_{i}\right)  \neq0.
\]
Thus the direction follows immediately.}

\item \revise{$\bigvee\limits_{\sigma\in S_{p}}\ \bigwedge\limits_{i=1}^{n}F^{\left(
\sigma_{i}\right)  }\left(  \alpha_{i}\right)  \neq0\ \ \Longrightarrow
\ \ \bigvee\limits_{\sigma\in S_{p}}\bigwedge\limits_{i=1}^{n}%
\operatorname*{mult}\left(  \alpha_{i}\right)  =\sigma_{i}\ \ $}

\revise{It is immediate from the following two sub-claims.}
\begin{enumerate}
\item \revise{$\bigwedge\limits_{i=1}^{n}F^{\left(  \sigma_{i}\right)  }\left(
\alpha_{i}\right)  \neq0$ $\ \Longrightarrow\ \bigwedge\limits_{i=1}%
^{n}\operatorname*{mult}\left(  \alpha_{i}\right)  \leq\sigma_{i}$}

\revise{It is immediate from elementary calculus.}

\item $\bigwedge\limits_{i=1}^{n}\operatorname*{mult}\left(  \alpha
_{i}\right)  \leq\sigma_{i}\Longrightarrow\ \ \bigwedge\limits_{i=1}%
^{n}\operatorname*{mult}\alpha_{i}=\sigma_{i}$

\revise{It is immediate from 
$\underset{\sigma\in Sp}{\ \forall}$ $\sum\limits_{i=1}^{n}%
\operatorname*{mult}\left(  \alpha_{i}\right)  =\sum\limits_{i=1}^{n}%
\sigma_{i}$, which is again obvious from the definitions of $\operatorname*{mult}$ and $S_{p}$.}
\end{enumerate}
\end{enumerate}

\item 
\revise{We rewrite the condition so that a  fewer polynomials are involved.
\[ \bigvee\limits_{\sigma\in S_{p}}\ \bigwedge\limits_{i=1}%
^{n}F^{\left(  \sigma_{i}\right)  }\left(  \alpha_{i}\right)  \neq
0\ \   \Longleftrightarrow\ \ \bigvee\limits_{\sigma\in S_{p}}%
\ \ \prod\limits_{i=1}^{n}F^{\left(  \sigma_{i}\right)  }\left(  \alpha
_{i}\right)  \neq0\]}

\revise{Proof: Obvious.}

\item \revise{We rewrite the condition so that only one polynomial is involved. 
\[ \bigvee\limits_{\sigma\in S_{p}}\ \ \prod\limits_{i=1}^{n}F^{\left(
\sigma_{i}\right)  }\left(  \alpha_{i}\right)  \neq0\ \ \Longleftrightarrow
\sum\limits_{\sigma\in S_{p}}\prod\limits_{i=1}^{n}F^{\left(  \sigma
_{i}\right)  }\left(  \alpha_{i}\right)  \ \neq0\]}

\revise{Proof: We prove the implication for both directions.}

\begin{enumerate}
\item \revise{$\ \bigvee\limits_{\sigma\in S_{p}}\ \ \prod\limits_{i=1}^{n}F^{\left(
\sigma_{i}\right)  }\left(  \alpha_{i}\right)  \neq0\ \ \Longrightarrow
\ \ \sum\limits_{\sigma\in S_{p}}\prod\limits_{i=1}^{n}F^{\left(  \sigma
_{i}\right)  }\left(  \alpha_{i}\right)  \ \neq0$}

\revise{Immediate from the fact  that if $\bigwedge\limits_{i=1}^{n}\operatorname*{mult}\left(
\alpha_{i}\right)  =\sigma_{i}$  then
$\underset{\pi\in S_p, \pi\neq \sigma}{\forall }~\prod_{i=1}%
^{n}F^{\left(  \pi_{i}\right)  }\left(  \alpha_{i}\right)  =0.$}

\item \revise{$\sum\limits_{\sigma\in S_{p}}\prod\limits_{i=1}^{n}F^{\left(
\sigma_{i}\right)  }\left(  \alpha_{i}\right)  \ \neq0\ \ \Longrightarrow
\ \ \bigvee\limits_{\sigma\in S_{p}}\ \ \prod\limits_{i=1}^{n}F^{\left(
\sigma_{i}\right)  }\left(  \alpha_{i}\right)  \neq0$.}

\revise{Obvious.}
\end{enumerate}

\revise{Now we have arrived at our goal by deriving a single symmetric polynomial
condition from the multiplicity condition given at the beginning.}

\revise{We will carry out a few ``cosmetic'' rewritings: (1) remove some redundancies and (2)\ write the condition more
compactly by recalling permanent.}

\item \revise{
We remove some redundancies in  the coefficients of $F^{\left(  \sigma
_{i}\right)  }\left(  \alpha_{i}\right)$. 
\[
\ \sum\limits_{\sigma\in S_{p}}\prod\limits_{i=1}^{n}F^{\left(  \sigma
_{i}\right)  }\left(  \alpha_{i}\right)  \ \neq0\Longleftrightarrow
\sum\limits_{\sigma\in S_{p}}\prod\limits_{i=1}^{n}F^{\left(  \sigma
_{i}\right)  }\left(  \alpha_{i}\right)  \ /\sigma_{i}!\neq0
\]
Proof: Obvious.}

\item \revise{We rewrite the condition more
compactly by recalling permanent.}

\revise{
Consider the following permanent:%
\[
P := \operatorname*{per}\left[
\begin{array}
[c]{ccc}%
\dfrac{F^{(p_{1})}(\alpha_{1})}{p_{1}!} & \cdots & \dfrac{F^{(p_{1})}%
(\alpha_{n})}{p_{1}!}\\[8pt]%
\vdots &  & \vdots\\[8pt]%
\dfrac{F^{(p_{n})}(\alpha_{1})}{p_{n}!} & \cdots & \dfrac{F^{(p_{n})}%
(\alpha_{n})}{p_{n}!}%
\end{array}
\right]
\]
Expanding $P$, we get
$$P=\sum_{\pi\in S_n}\prod_{i=1}^nF^{(p_{\pi(i)})}(\alpha_{i})\big/p_{\pi(i)}!$$}

\revise{
Since, for $\sigma\in S_p$, there are
$c=\prod_{i=1}^{\ell}q_{i}!$ distinct permutations $\pi$'s in $S_n$ such that $\pi(p)=\sigma$ where $q_{1},\ldots,q_{\ell}$ are
the occurrences of distinct numbers in $p$, we have
\[\sum_{\sigma\in S_p}\prod_{i=1}^nF^{(\sigma_i)}(\alpha_{i})\big/\sigma_i!=\operatorname*{per}\left[
\begin{array}
[c]{ccc}%
\dfrac{F^{(p_{1})}(\alpha_{1})}{p_{1}!} & \cdots & \dfrac{F^{(p_{1})}%
(\alpha_{n})}{p_{1}!}\\[8pt]%
\vdots &  & \vdots\\[8pt]%
\dfrac{F^{(p_{n})}(\alpha_{1})}{p_{n}!} & \cdots & \dfrac{F^{(p_{n})}%
(\alpha_{n})}{p_{n}!}%
\end{array}
\right]/c\]
}
\item \revise{
By  denoting  $\operatorname*{per}%
\left[
\begin{array}
[c]{ccc}%
\dfrac{F^{(p_{1})}(\alpha_{1})}{p_{1}!} & \cdots & \dfrac{F^{(p_{1})}%
(\alpha_{n})}{p_{1}!}\\[8pt]%
\vdots &  & \vdots\\[8pt]%
\dfrac{F^{(p_{n})}(\alpha_{1})}{p_{n}!} & \cdots & \dfrac{F^{(p_{n})}%
(\alpha_{n})}{p_{n}!}%
\end{array}
\right]  /c$ by $\overline{D}_{\mu}\left(  F\right.)$, we finally have
\[
\operatorname*{mult}\left(  F\right)  =\mu\ \ \Longleftrightarrow
\ \ \overline{D}_{\mu}\left(  F\right)  \neq0
\]
}
\end{enumerate}
\end{proof}

\begin{example}
\label{ex:multdisr} Let
\revise{$F=(x-\alpha_1)(x-\alpha_2)(x-\alpha_3)(x-\alpha_4)$ and}
$\mu=(3,1)$. Then $p=(3,3,3,1)$ and $c= 3!\cdot1!$. Thus
\begin{align*}
\overline{D}_{\mu}(F) =  &  \,\operatorname*{per}\left[
\begin{array}
[c]{cccc}%
\dfrac{F^{(3)}(\alpha_{1})}{3!} & \dfrac{F^{(3)}(\alpha_{2})}{3!} &
\dfrac{F^{(3)}(\alpha_{3})}{3!} & \dfrac{F^{(3)}(\alpha_{4})}{3!}\\[8pt]%
\dfrac{F^{(3)}(\alpha_{1})}{3!} & \dfrac{F^{(3)}(\alpha_{2})}{3!} &
\dfrac{F^{(3)}(\alpha_{3})}{3!} & \dfrac{F^{(3)}(\alpha_{4})}{3!}\\[8pt]%
\dfrac{F^{(3)}(\alpha_{1})}{3!} & \dfrac{F^{(3)}(\alpha_{2})}{3!} &
\dfrac{F^{(3)}(\alpha_{3})}{3!} & \dfrac{F^{(3)}(\alpha_{4})}{3!}\\[8pt]%
\dfrac{F^{(1)}(\alpha_{1})}{1!} & \dfrac{F^{(1)}(\alpha_{2})}{1!} &
\dfrac{F^{(1)}(\alpha_{3})}{1!} & \dfrac{F^{(1)}(\alpha_{4})}{1!}%
\end{array}
\right]  /(3!\cdot1!)\\
=  &  \,{\operatorname*{per}\left[
\revise{\begin{array} [c]{cccc}\sum\limits_{\substack{1\le i\le 4\\i\ne1}}(\alpha_{1}-\alpha_{i}) & \sum\limits_{\substack{1\le i\le 4\\i\ne2}}(\alpha_{2}-\alpha_{i}) & \sum\limits_{\substack{1\le i\le 4\\i\ne3}}(\alpha_{3}-\alpha_{i}) & \sum\limits_{\substack{1\le i\le 4\\i\ne4}}(\alpha_{4}-\alpha_{i})\\[8pt] \sum\limits_{\substack{1\le i\le 4\\i\ne1}}(\alpha_{1}-\alpha_{i}) & \sum\limits_{\substack{1\le i\le 4\\i\ne2}}(\alpha_{2}-\alpha_{i}) & \sum\limits_{\substack{1\le i\le 4\\i\ne3}}(\alpha_{3}-\alpha_{i}) & \sum\limits_{\substack{1\le i\le 4\\i\ne4}}(\alpha_{4}-\alpha_{i})\\[8pt] \sum\limits_{\substack{1\le i\le 4\\i\ne1}}(\alpha_{1}-\alpha_{i}) & \sum\limits_{\substack{1\le i\le 4\\i\ne2}}(\alpha_{2}-\alpha_{i}) & \sum\limits_{\substack{1\le i\le 4\\i\ne3}}(\alpha_{3}-\alpha_{i}) & \sum\limits_{\substack{1\le i\le 4\\i\ne4}}(\alpha_{4}-\alpha_{i})\\[8pt]\prod\limits_{\substack{1\le i\le 4\\i\ne1}}(\alpha_{1}-\alpha_{i}) & \prod\limits_{\substack{1\le i\le 4\\i\ne2}}(\alpha_{2}-\alpha_{i}) & \prod\limits_{\substack{1\le i\le 4\\i\ne3}}(\alpha_{3}-\alpha_{i}) & \prod\limits_{\substack{1\le i\le 4\\i\ne4}}(\alpha_{4}-\alpha_{i}) \end{array}}
\right]  }/(3!\cdot1!)\\
=  &  \,-{}(\alpha_{1}+\alpha_{2}-\alpha_{3}-\alpha_{4})^{2}(\alpha_{1}%
+\alpha_{3}-\alpha_{2}-\alpha_{4})^{2}(\alpha_{1}+\alpha_{4}-\alpha_{2}%
-\alpha_{3})^{2}.
\end{align*}
If we know that $F$ has two distinct roots, then $\overline{D}_{\mu}(F)\ne0$
if and only if $\mu=(3,1)$.
\end{example}

\begin{remark}
\label{rm:sylvester}
\revise{Suppose $F=a_n\prod_{i=1}^n(x-\alpha_i)$.} Let $\mu
=(1,\ldots,1)$. Then
\[
\overline{D}_{\mu}(F)=\prod_{i=1}^{n}F^{\prime}(\alpha_{i})={a_{{n}}^{n-1}%
}\prod_{i\neq j}\left(  \alpha_{i}-\alpha_{j}\right)
\]
which is the well known discriminant up to sign. Thus $\overline{D}_{\mu}(F)$
in Lemma \ref{lem:expr-in-roots} can be viewed as a certain generalization of discriminant.
\end{remark}

\subsection{From a permanent in roots to a sum of determinants in roots}

The results presented in this subsection and the next subsection are more
general than what are needed for proving the main result (Theorem
\ref{thm:multdiscriminant}). We present the more general results in the hope
that they would be useful for some other related problems. The following lemma
shows that one can rewrite a permanent in terms of determinants.

\begin{lemma}
\label{lem:gen-D-ratio}
\revise{Let $A$ and $B$ be square matrices of size $n$.} We have
\[
\operatorname*{per}\left(  A\right)  =\frac{1}{\det\left(  B\right)  }%
\sum_{\tau\in S_{n}}\det\left(  \left(  P_{\tau}A\right)  \circ B\right)
\]
where

\begin{itemize}
\item
\revise{The notation $\circ$ stands for the entry-wise (Hadamard) product; in other words,
the $(i,j)$-th entry of $A\circ B$ is the product of the $(i,j)$-th entries of $A$ and $B$;}

\item \revise{The notation $P_{\tau}B$ stands for the matrix obtained by permuting the rows of $B$ as indicated by $\tau$.}
\end{itemize}
\end{lemma}

\begin{proof}
\revise{We will rewrite a permanent in terms of determinants as follows.}

\begin{enumerate}
\item \revise{Recalling the definition of permanent, we have%
\[
\operatorname*{per}\left(  A\right)  =\sum_{\tau\in S_{n}}\ \prod_{j=1}%
^{n}a_{\tau\left(  j\right)  ,j}\
\]
}

\item \revise{
Now we make a simple, but crucial rewriting of the above expression into
the following%
\[
\operatorname*{per}\left(  A\right)  =\sum_{\tau\in S_{n}}\ \prod_{j=1}%
^{n}a_{\left(  \tau\circ\pi\right)  \left(  j\right)  ,j}%
\]
for arbitrary $\pi\in S_{n}$. Note that $\tau\ $is replaced with
$\tau\circ\pi$.  This is correct because  $\tau\circ\pi$ also ranges over~$S_{n}$. Why we make the above rewriting will be made clear in the following steps.}

\item \revise{
Recalling the definition of determinant, we have%
\[
\operatorname*{per}\left(  A\right)  \det\left(  B\right)  =\left(  \sum
_{\tau\in S_{n}}\ \prod_{j=1}^{n}a_{\tau\left(  \pi\left(  j\right)  \right)
,j}\right)  \left(  \sum_{\pi\in S_{n}}\ \operatorname*{sgn}\left(
\pi\right)  \prod_{j=1}^{n}b_{\pi\left(  j\right)  ,j}\right)
\]
}

\item \revise{Rearranging the sums and the products, we have%
\[
\operatorname*{per}\left(  A\right)  \det\left(  B\right)  =\sum_{\tau\in
S_{n}}\sum_{\pi\in S_{n}}\ \operatorname*{sgn}\left(  \pi\right)  \prod
_{j=1}^{n}\left(  a_{\tau\left(  \pi\left(  j\right)  \right)  ,j}%
b_{\pi\left(  j\right)  ,j}\right)
\]
}

\item \revise{Writing in terms of determinants and Hadamard product, we  have
\[
\operatorname*{per}\left(  A\right)  \det\left(  B\right)  =\sum_{\tau\in
S_{n}}\det\left(  \left(  P_{\tau}A\right)  \circ B\right)
\]
}

\item \revise{
Finally we have%
\[
\operatorname*{per}\left(  A\right)  =\frac{1}{\det\left(  B\right)  }%
\sum_{\tau\in S_{n}}\det\left(  \left(  P_{\tau}A\right)  \circ B\right)
\]}
\end{enumerate}
\end{proof}

\begin{example}
When $n=2$, we have $S_{2}=\{(1),(12)\}$. Let $A=\left[
\begin{array}
[c]{cc}%
a_{11} & a_{12}\\
a_{21} & a_{22}%
\end{array}
\right]$ and $B=\left[
\begin{array}
[c]{cc}%
b_{11} & b_{12}\\
b_{21} & b_{22}%
\end{array}
\right]$. Then \revise{
we have
\begin{align*}
&\operatorname*{per}\left(A\right)=a_{11}a_{22}+a_{12}a_{21},\qquad
\det\left(B\right)=b_{11}b_{22}-b_{12}b_{21}\\[10pt]
&\sum_{\tau\in S_{2}}\det\left(  \left(  P_{\tau}A\right)  \circ B\right)\\
=&\det\left(P_{(1)}A\circ B\right)+\det\left(P_{(12)}A\circ B\right) \\
=&\det\left[
\begin{array}
[c]{cc}%
a_{11}b_{11} & a_{12}b_{12}\\
a_{21}b_{21} & a_{22}b_{22}%
\end{array}
\right]  +\det\left[
\begin{array}
[c]{cc}%
a_{21}b_{11} & a_{22}b_{12}\\
a_{11}b_{21} & a_{12}b_{22}%
\end{array}
\right]  \\
=&(a_{11}a_{22}b_{11}b_{22}-a_{12}a_{21}b_{12}b_{21})+(%
a_{21}a_{12}b_{11}b_{22}-a_{11}a_{22}b_{12}b_{21})\\
=&(a_{11}a_{22}+a_{12}a_{21})(b_{11}b_{22}-b_{12}b_{21})\\
=&\operatorname*{per}\left(A\right)\det\left(B\right)
\end{align*}
Thus, 
$$\operatorname*{per}\left(A\right)=\frac{1}{\det\left(  B\right)  }
\sum_{\tau\in S_{2}}\det\left(  \left(  P_{\tau}A\right)  \circ B\right)$$}
\end{example}

\subsection{From a determinant in roots to a determinant in coefficients}

{\revise{Let $\mathbb{C}\left[  x\right]  _{k}$ stand for  the set of all  polynomials in $\mathbb{C}\left[  x\right]$ with degree at most $k$.}
 
\begin{lemma}
\label{lem:ratio} Let $\omega_{0},\ldots,\omega_{k}$ be a
\revise{canonical} basis of $\mathbb{C}\left[  x\right]_{k}$ for every $k\geq0$.
\revise{Let $F=a_n(x-\alpha_1)\cdots(x-\alpha_n)$ and $G_{1},\ldots
,G_{n}\in\mathbb{C}\left[  x\right]_{2n-2}$.} Then we have%
\[
\operatorname*{dp}\left[
\begin{array}
[c]{c}%
\omega_{n-2}F\\
\vdots\\
\omega_{0}F\\
G_{1}\\
\vdots\\
G_{n}%
\end{array}
\right]  =\dfrac{a_{n}^{n-1}\cdot\det\left[
\begin{array}
[c]{ccc}%
G_{1}(\alpha_{1}) & \cdots & G_{1}(\alpha_{n})\\
\vdots &  & \vdots\\
G_{n}(\alpha_{1}) & \cdots & G_{n}\left(  \alpha_{n}\right)
\end{array}
\right]  }{\det\left[
\begin{array}
[c]{ccc}%
\omega_{n-1}\left(  \alpha_{1}\right)  & \cdots & \omega_{n-1}\left(
\alpha_{n}\right) \\
\vdots &  & \vdots\\
\omega_{0}\left(  \alpha_{1}\right)  & \cdots & \omega_{0}\left(  \alpha
_{n}\right)
\end{array}
\right]  }%
\]

\end{lemma}

\begin{proof}
We will derive the expression step by step.

\begin{enumerate}
\item Let%
\[
M_{F}=\left[
\begin{array}
[c]{c}%
\omega_{n-2}F\\
\vdots\\
\omega_{0}F
\end{array}
\right]  \ \ \ \ \ \ M_{G}=\left[
\begin{array}
[c]{c}%
G_{1}\\
\vdots\\
G_{n}%
\end{array}
\right]
\]

\item Let $M\in\mathbb{C}^{\left(  2n-1\right)  \times\left(  2n-1\right)  }$
be such that%
\[
\left[
\begin{array}
[c]{c}%
M_{F}\\
M_{G}%
\end{array}
\right]  =M\left[
\begin{array}
[c]{c}%
\omega_{2n-2}\\
\vdots\\
\omega_{0}%
\end{array}
\right]
\]
Then by the definition of $\operatorname*{dp}$, we have $\operatorname*{dp}%
\left[
\begin{array}
[c]{c}%
M_{F}\\
M_{G}%
\end{array}
\right]  =\det\left( M\right)  $.

\item Let us partition $M$ naturally as%
\[
M=\left[
\begin{array}
[c]{cc}%
A & B\\
C & D
\end{array}
\right]
\]
where%
\[%
\begin{array}
[c]{ll}%
A\in\mathbb{C}^{\left(  n-1\right)  \times\left(  n-1\right)  } &
B\in\mathbb{C}^{\left(  n-1\right)  \times n}\\
C\in\mathbb{C}^{n\times\left(  n-1\right)  } & D\in\mathbb{C}^{n\times n}%
\end{array}
\]

\item Now we introduce a crucial object in the derivation.%
\[
W=\left[
\begin{array}
[c]{cc}%
I_{n-1} & U\\
& V
\end{array}
\right]
\]
where
\[
U=\left[
\begin{array}
[c]{ccc}%
\omega_{2n-2}\left(  \alpha_{1}\right)  & \cdots & \omega_{2n-2}\left(
\alpha_{n}\right) \\
\vdots &  & \vdots\\
\omega_{n}\left(  \alpha_{1}\right)  & \cdots & \omega_{n}\left(  \alpha
_{n}\right)
\end{array}
\right]  \ \ \ \text{and }\ \ V=\left[
\begin{array}
[c]{ccc}%
\omega_{n-1}\left(  \alpha_{1}\right)  & \cdots & \omega_{n-1}\left(
\alpha_{n}\right) \\
\vdots &  & \vdots\\
\omega_{0}\left(  \alpha_{1}\right)  & \cdots & \omega_{0}\left(  \alpha
_{n}\right)
\end{array}
\right]
\]
Note that $V$ is the generalized Vandermonde matrix of $\revise{F}$
(up to ordering of rows). A similar object was also used in
~\cite{2007_Andrea_Hong_Krick_Szanto} for studying Sylvester double sum.

\item Note%
\begin{align*}
MW  &  =\left[
\begin{array}
[c]{cc}%
A & B\\
C & D
\end{array}
\right]  \left[
\begin{array}
[c]{cc}%
I_{n-1} & U\\
0 & V
\end{array}
\right] \\
&  =\left[
\begin{array}
[c]{cc}%
A & AU+BV\\
C & CU+DV
\end{array}
\right] \\
&  =\left[
\begin{array}
[c]{cccc}%
A & M_{F}\left(  \alpha_{1}\right)  & \cdots & M_{F}\left(  \alpha_{n}\right)
\\
C & M_{G}\left(  \alpha_{1}\right)  & \cdots & M_{G}\left(  \alpha_{n}\right)
\end{array}
\right] \\
&  =\left[
\begin{array}
[c]{cccc}%
A &  &  & \\
C & M_{G}\left(  \alpha_{1}\right)  & \cdots & M_{G}\left(  \alpha_{n}\right)
\end{array}
\right]  \ \ \ \text{since }F\left(  \alpha_{i}\right)  =0
\end{align*}

\item Thus%
\begin{align*}
\det\left(  M\right)  \cdot\det\left(  W\right)  &=
\det\left(  A\right)
\cdot\det\left[
\begin{array}
[c]{ccc}%
M_{G}\left(  \alpha_{1}\right)  & \cdots & M_{G}\left(  \alpha_{n}\right)
\end{array}
\right]\\
&=\det\left(  A\right)
\cdot\det\left[
\begin{array}
[c]{ccc}%
G\left(  \alpha_{1}\right)  & \cdots & G\left(  \alpha_{n}\right)
\end{array}
\right]
\end{align*}
\revise{where $G(\alpha_i)=\left[
\begin{array}
[c]{ccc}G_1\left(  \alpha_{i}\right)  & \cdots & G_n\left(  \alpha_{i}\right)
\end{array}
\right]^T$.}

\item Note that $\det\left(  W\right)  =\det\left(  I_{n-1}\right)  \cdot
\det\left(  V\right)  =\det\left(  V\right)  $ since $W$ is block-triangular.

\item Note that $\det\left(  A\right)  =a_{n}^{n-1}$ since $A$ is triangular
and the diagonal elements are $a_{n}$.

\item By putting together we have $\operatorname*{dp}\left[
\begin{array}
[c]{c}%
M_{F}\\
M_{G}%
\end{array}
\right]  =\det\left(  M\right)  =\dfrac{a_{n}^{n-1}\cdot\det\left[
\begin{array}
[c]{ccc}%
G\left(  \alpha_{1}\right)  & \cdots & G\left(  \alpha_{n}\right)
\end{array}
\right]  }{\det\left(  V\right)  }$
\end{enumerate}
\end{proof}

\begin{example}
Let $F=a_{3}x^{3}+a_{2}x^{2}+a_{1}x+a_{0}=a_{3}(x-\alpha_{1})(x-\alpha
_{2})(x-\alpha_{3})$ and $G_{i}=x^{2}F^{(i)}/i!$. Let $\omega_{i}=x^{i}$.
Then
\begin{align*}
\operatorname*{dp}\left[
\begin{array}
[c]{c}%
\omega_{1}F\\
\omega_{0}F\\
G_{1}\\
G_{2}\\
G_{3}%
\end{array}
\right]   &  =\det\left[
\begin{array}
[c]{ccccc}%
a_{3} & a_{2} & a_{1} & a_{0} & \\
0 & a_{3} & a_{2} & a_{1} & a_{0}\\
3a_{3} & 2\,a_{2} & a_{1} &  & \\
& 3a_{3} & a_{2} &  & \\
&  & a_{3} &  &
\end{array}
\right]  =9\,a_{3}^{3}a_{0}^{2}%
\end{align*}

\begin{align*}
& \det{\left[
\begin{array}
[c]{ccc}%
G_{1}(\alpha_{1}) & G_{1}\left(  \alpha_{2}\right)  & G_{1}(\alpha_{3})\\
G_{2}(\alpha_{1}) & G_{2}\left(  \alpha_{2}\right)  & G_{2}(\alpha_{3})\\
G_{3}(\alpha_{1}) & G_{3}\left(  \alpha_{2}\right)  & G_{3}(\alpha_{3})
\end{array}
\right]  }\\
=  &  \det\left[
\begin{array}
[c]{ccc}%
a_{3}{\alpha_{{1}}}^{2}\left(  \alpha_{{1}}-\alpha_{{2}}\right)  \left(
\alpha_{{1}}-\alpha_{{3}}\right)  & a_{3}{\alpha_{{2}}}^{2}\left(  \alpha
_{{2}}-\alpha_{{1}}\right)  \left(  \alpha_{{2}}-\alpha_{{3}}\right)  &
a_{3}{\alpha_{{3}}}^{2}\left(  \alpha_{{3}}-\alpha_{{1}}\right)  \left(
\alpha_{{3}}-\alpha_{{2}}\right) \\[4pt]%
\noalign{\medskip}a_{3}{\alpha_{{1}}}^{2}\left(  2\,\alpha_{{1}}-\alpha_{{3}%
}-\alpha_{{2}}\right)  & a_{3}{\alpha_{{2}}}^{2}\left(  2\,\alpha_{{2}}%
-\alpha_{{3}}-\alpha_{{1}}\right)  & a_{3}{\alpha_{{3}}}^{2}\left(
2\,\alpha_{{3}}-\alpha_{{2}}-\alpha_{{1}}\right) \\
\noalign{\medskip}a_{3}{\alpha_{{1}}}^{2} & a_{3}{\alpha_{{2}}}^{2} &
a_{3}{\alpha_{{3}}}^{2}%
\end{array}
\right] \\[4pt]
=  &  9\,a_{3}^{3}{\alpha_{{1}}}^{2}{\alpha_{{2}}}^{2}{\alpha_{{3}}}%
^{2}\left(  \alpha_{{2}}-\alpha_{{3}}\right)  \left(  \alpha_{{1}}-\alpha
_{{3}}\right)  \left(  \alpha_{{1}}-\alpha_{{2}}\right) \\[4pt]
=  &  9\,a_{3}^{1}a_{0}^{2}\left(  \alpha_{{2}}-\alpha_{{3}}\right)  \left(
\alpha_{{1}}-\alpha_{{3}}\right)  \left(  \alpha_{{1}}-\alpha_{{2}}\right)
\\[8pt]
&\det{\left[
\begin{array}
[c]{ccc}%
\omega_{2}\left(  \alpha_{1}\right)  & \omega_{2}\left(  \alpha_{2}\right)  &
\omega_{2}\left(  \alpha_{3}\right) \\
\omega_{1}\left(  \alpha_{1}\right)  & \omega_{1}\left(  \alpha_{2}\right)  &
\omega_{1}\left(  \alpha_{3}\right) \\
\omega_{0}\left(  \alpha_{1}\right)  & \omega_{0}\left(  \alpha_{2}\right)  &
\omega_{0}\left(  \alpha_{3}\right)
\end{array}
\right]  }\ =\det\left[
\begin{array}
[c]{ccc}%
\alpha_{1}^{2} & \alpha_{2}^{2} & \alpha_{3}^{2}\\
\alpha_{1}^{1} & \alpha_{2}^{1} & \alpha_{3}^{1}\\
\alpha_{1}^{0} & \alpha_{2}^{0} & \alpha_{3}^{0}%
\end{array}
\right]  =(\alpha_{1}-\alpha_{2})(\alpha_{1}-\alpha_{3})(\alpha_{2}-\alpha
_{3})
\end{align*}
Thus we have
\[
\operatorname*{dp}\left[
\begin{array}
[c]{c}%
\omega_{1}F\\
\omega_{0}F\\
G_{1}\\
G_{2}\\
G_{3}%
\end{array}
\right]  =\frac{a_{3}^{2}\cdot\det{\left[
\begin{array}
[c]{ccc}%
G_{1}(\alpha_{1}) & G_{1}\left(  \alpha_{2}\right)  & G_{1}(\alpha_{3})\\
G_{2}(\alpha_{1}) & G_{2}\left(  \alpha_{2}\right)  & G_{2}(\alpha_{3})\\
G_{3}(\alpha_{1}) & G_{3}\left(  \alpha_{2}\right)  & G_{3}(\alpha_{3})
\end{array}
\right]  }}{{\det\left[
\begin{array}
[c]{ccc}%
\omega_{2}\left(  \alpha_{1}\right)  & \omega_{2}\left(  \alpha_{2}\right)  &
\omega_{2}\left(  \alpha_{3}\right) \\
\omega_{1}\left(  \alpha_{1}\right)  & \omega_{1}\left(  \alpha_{2}\right)  &
\omega_{1}\left(  \alpha_{3}\right) \\
\omega_{0}\left(  \alpha_{1}\right)  & \omega_{0}\left(  \alpha_{2}\right)  &
\omega_{0}\left(  \alpha_{3}\right)
\end{array}
\right]  }}%
\]

\end{example}

\subsection{Proof of Main Result (Theorem \ref{thm:multdiscriminant})}

Finally we will prove the main result by combining the above three lemmas
(Lemma \ref{lem:expr-in-roots}, \ref{lem:gen-D-ratio} and \ref{lem:ratio}).

\begin{enumerate}
\item $\,$From Lemma \ref{lem:expr-in-roots}, we have
\[
\operatorname*{mult}\left(  F\right)  =\mu\ \ \Longleftrightarrow
\operatorname*{per}\left[
\begin{array}
[c]{ccc}%
\dfrac{F^{(p_{1})}(\alpha_{1})}{p_{1}!} & \cdots & \dfrac{F^{(p_{1})}%
(\alpha_{n})}{p_{1}!}\\[8pt]%
\vdots &  & \vdots\\[8pt]%
\dfrac{F^{(p_{n})}(\alpha_{1})}{p_{n}!} & \cdots & \dfrac{F^{(p_{n})}%
(\alpha_{n})}{p_{n}!}%
\end{array}
\right]  \neq0
\]

\item Applying Lemma \ref{lem:gen-D-ratio} to
\[
A=\left[
\begin{array}
[c]{ccc}%
\alpha_{1}^{n-1} & \cdots & \alpha_{n}^{n-1}\\
\vdots &  & \vdots\\
\alpha_{1}^{0} & \cdots & \alpha_{n}^{0}%
\end{array}
\right]  ,\quad B=\left[
\begin{array}
[c]{ccc}%
\dfrac{F^{(p_{1})}(\alpha_{1})}{p_{1}!} & \cdots & \dfrac{F^{(p_{1})}%
(\alpha_{n})}{p_{1}!}\\[8pt]%
\vdots &  & \vdots\\[8pt]%
\dfrac{F^{(p_{n})}(\alpha_{1})}{p_{n}!} & \cdots & \dfrac{F^{(p_{n})}%
(\alpha_{n})}{p_{n}!}%
\end{array}
\right]
\]
and dividing $\det\left(  A\right)  $ on both sides, we have
\[
\setlength{\arraycolsep}{4.2pt}\operatorname*{per}\left[
\begin{array}
[c]{ccc}%
\dfrac{F^{(p_{1})}(\alpha_{1})}{p_{1}!} & \cdots & \dfrac{F^{(p_{1})}%
(\alpha_{n})}{p_{1}!}\\[8pt]%
\vdots &  & \vdots\\[8pt]%
\dfrac{F^{(p_{n})}(\alpha_{1})}{p_{n}!} & \cdots & \dfrac{F^{(p_{n})}%
(\alpha_{n})}{p_{n}!}%
\end{array}
\right]  =\,\dfrac{\sum\limits_{\tau\in S_{n}}\det\left(\left[
\begin{array}
[c]{ccc}%
\alpha_{1}^{n-1} & \cdots & \alpha_{n}^{n-1}\\
\vdots &  & \vdots\\
\alpha_{1}^{0} & \cdots & \alpha_{n}^{0}%
\end{array}
\right]  \circ P_{\tau}\left[
\begin{array}
[c]{ccc}%
\dfrac{F^{(p_{1})}(\alpha_{1})}{p_{1}!} & \cdots & \dfrac{F^{(p_{1})}%
(\alpha_{n})}{p_{1}!}\\[8pt]%
\vdots &  & \vdots\\[8pt]%
\dfrac{F^{(p_{n})}(\alpha_{1})}{p_{n}!} & \cdots & \dfrac{F^{(p_{n})}%
(\alpha_{n})}{p_{n}!}%
\end{array}
\right]\right)  }{\operatorname*{det}\left[
\begin{array}
[c]{ccc}%
\alpha_{1}^{n-1} & \cdots & \alpha_{n}^{n-1}\\
\vdots &  & \vdots\\
\alpha_{1}^{0} & \cdots & \alpha_{n}^{0}%
\end{array}
\right]  }%
\]
Since for $\sigma\in S_{p},\text{there are }c=\prod_{i=1}^{\ell}q_{i}%
!~\tau\text{'s such that }\tau(p)=\sigma$ where $q_{1},\ldots,q_{\ell}$ are
the occurrences of distinct numbers in $p$, we have
\begin{align*}
\setlength{\arraycolsep}{1.5pt}\operatorname*{per}\left[
\begin{array}
[c]{ccc}%
\dfrac{F^{(p_{1})}(\alpha_{1})}{p_{1}!} & \cdots & \dfrac{F^{(p_{1})}%
(\alpha_{n})}{p_{1}!}\\[8pt]%
\vdots &  & \vdots\\[8pt]%
\dfrac{F^{(p_{n})}(\alpha_{1})}{p_{n}!} & \cdots & \dfrac{F^{(p_{n})}%
(\alpha_{n})}{p_{n}!}%
\end{array}
\right]   &  =c\cdot\,\dfrac{\sum\limits_{\sigma\in S_{p}}\operatorname*{det}%
\left(  \left[
\begin{array}
[c]{ccc}%
\alpha_{1}^{n-1} & \cdots & \alpha_{n}^{n-1}\\
\vdots &  & \vdots\\
\alpha_{1}^{0} & \cdots & \alpha_{n}^{0}%
\end{array}
\right]  \circ\left[
\begin{array}
[c]{ccc}%
\dfrac{F^{(\sigma_{1})}(\alpha_{1})}{\sigma_{1}!} & \cdots & \dfrac
{F^{(\sigma_{1})}(\alpha_{n})}{\sigma_{1}!}\\[8pt]%
\vdots &  & \vdots\\[8pt]%
\dfrac{F^{(\sigma_{n})}(\alpha_{1})}{\sigma_{n}!} & \cdots & \dfrac
{F^{(\sigma_{n})}(\alpha_{n})}{\sigma_{n}!}%
\end{array}
\right]  \right)  }{\operatorname*{det}\left[
\begin{array}
[c]{ccc}%
\alpha_{1}^{n-1} & \cdots & \alpha_{n}^{n-1}\\
\vdots &  & \vdots\\
\alpha_{1}^{0} & \cdots & \alpha_{n}^{0}%
\end{array}
\right]  }\\
&  =\,c\cdot\sum\limits_{\sigma\in S_{p}}\frac{\operatorname*{det}\left[
\begin{array}
[c]{ccc}%
\alpha_{1}^{n-1}\dfrac{F_{1}^{\left(  \sigma_{1}\right)  }\left(  \alpha
_{1}\right)  }{\sigma_{1}!} & \cdots & \alpha_{n}^{n-1}\dfrac{F_{n}^{\left(
\sigma_{1}\right)  }\left(  \alpha_{n}\right)  }{\sigma_{1}!}\\[8pt]%
\vdots &  & \vdots\\[8pt]%
\alpha_{1}^{0}\dfrac{F_{1}^{\left(  \sigma_{n}\right)  }\left(  \alpha
_{1}\right)  }{\sigma_{n}!} & \cdots & \alpha_{n}^{0}\dfrac{F_{n}^{\left(
\sigma_{n}\right)  }\left(  \alpha_{n}\right)  }{\sigma_{n}!}%
\end{array}
\right]  }{\operatorname*{det}\left[
\begin{array}
[c]{ccc}%
\alpha_{1}^{n-1} & \cdots & \alpha_{n}^{n-1}\\
\vdots &  & \vdots\\
\alpha_{1}^{0} & \cdots & \alpha_{n}^{0}%
\end{array}
\right]  }%
\end{align*}

\item By applying Lemma \ref{lem:ratio} to $G_{i}(x)=\dfrac{x^{n-i}%
F^{(\sigma_{i})}(x)}{\sigma_{i}!}$ and $\omega_{i}=x^{i}$, we have
\[
\dfrac{a_{n}^{n-1}\cdot\operatorname*{det}\left[
\begin{array}
[c]{ccc}%
\alpha_{1}^{n-1}\dfrac{F_{{}}^{\left(  \sigma_{1}\right)  }\left(  \alpha
_{1}\right)  }{\sigma_{1}!} & \cdots & \alpha_{n}^{n-1}\dfrac{F^{\left(
\sigma_{1}\right)  }\left(  \alpha_{n}\right)  }{\sigma_{1}!}\\[8pt]%
\vdots &  & \vdots\\[8pt]%
\alpha_{1}^{0}\dfrac{F^{\left(  \sigma_{n}\right)  }\left(  \alpha_{1}\right)
}{\sigma_{n}!} & \cdots & \alpha_{n}^{0}\dfrac{F^{\left(  \sigma_{n}\right)
}\left(  \alpha_{n}\right)  }{\sigma_{n}!}%
\end{array}
\right]  }{\operatorname*{det}\left[
\begin{array}
[c]{ccc}%
\alpha_{1}^{n-1} & \cdots & \alpha_{n}^{n-1}\\
\vdots &  & \vdots\\
\alpha_{1}^{0} & \cdots & \alpha_{n}^{0}%
\end{array}
\right]  }=\operatorname*{dp}\left[
\begin{array}
[c]{c}%
x^{n-2}F\\
\vdots\\
x^{0}F\\
x^{n-1}F^{\left(  \sigma_{_{1}}\right)  }/\sigma_{1}!\\
\vdots\\
x^{0}F^{\left(  \sigma_{_{n}}\right)  }/\sigma_{n}!
\end{array}
\right]  .
\]

\end{enumerate}

\noindent Finally, combining the above three steps, we have
\[
\operatorname*{mult}\left(  F\right)  =\mu\ \ \Longleftrightarrow\ \ D_{\mu
}\left(  F\right)  \neq0.
\]
We have proved the main result (Theorem \ref{thm:multdiscriminant}).

\section{Comparison}

\label{sec:compare}In this section, we compare the \textquotedblleft
sizes\textquotedblright\ of the multiplicity-discriminanting condition in Theorem
\ref{thm:multdiscriminant} and that given by a complex root version of
Yang-Hou-Zeng (YHZ) \cite{1996_Yang_Hou_Zeng}.\footnote{They introduced the
key concept of complete discriminant system for polynomials which is a set of
explicit expressions of the coefficients to determine the numbers and
multiplicities of real and non-real roots. The conditions in the solution to
the complete root classification problem consists of equations, inequations
and inequalities. When restricted to discriminating only multiplicities of
roots (without discriminating between real and non-real roots), computing a
complete discriminant system is equivalent to computing greatest common
divisors iteratively.} Specifically we compare the number and the maximum
degrees of polynomials appearing in the conditions. In Table \ref{tbl:compare}%
, we show a comparison for $n=8$. They were determined through brute-force
computations. In the table, we used the following short-hands:

\begin{table}[t]
\caption{Comparison}%
\label{tbl:compare}
\begin{center}%
\begin{tabular}
[c]{|c|c|l||c|c||c|c|}\hline
$n$ & $m$ & $\mu$ & \#$_{\mathrm{NEW}}$ & \#$_{\mathrm{YHZ}}$ &
$d_{\mathrm{NEW}}$ & $d_{\mathrm{YHZ}}$\\\hline
\multirow{18}{*}{8} & \multirow{4}{*}{2} & $[4,4]$ & 1 & 7 & 12 &
81\\\cline{3-7}
&  & $[5,3]$ & 1 & 8 & 13 & 81\\\cline{3-7}
&  & $[6,2]$ & 1 & 11 & 14 & 63\\\cline{3-7}
&  & $[7,1]$ & 1 & 16 & 15 & 33\\\cline{2-3}\cline{3-7}
& \multirow{4}{*}{3} & $[3,3,2]$ & 1 & 3 & 14 & 75\\\cline{3-7}
&  & $[4,2,2]$ & 1 & 4 & 14 & 75\\\cline{3-7}
&  & $[4,3,1]$ & 1 & 5 & 15 & 75\\\cline{3-7}
&  & $[5,2,1]$ & 1 & 7 & 15 & 75\\\cline{3-7}
&  & $[6,1,1]$ & 1 & 11 & 15 & 45\\\cline{2-3}\cline{3-7}
& \multirow{4}{*}{4} & $[2,2,2,2]$ & 1 & 1 & 14 & 49\\\cline{3-7}
&  & $[3,2,2,1]$ & 1 & 2 & 15 & 49\\\cline{3-7}
&  & $[3,3,1,1]$ & 1 & 3 & 15 & 63\\\cline{3-7}
&  & $[4,2,1,1]$ & 1 & 4 & 15 & 63\\\cline{2-3}\cline{3-7}
& \multirow{3}{*}{5} & $[2,2,2,1,1]$ & 1 & 1 & 15 & 45\\\cline{3-7}
&  & $[3,2,1,1,1]$ & 1 & 2 & 15 & 45\\\cline{3-7}
&  & $[4,1,1,1,1]$ & 1 & 4 & 15 & 45\\\cline{2-3}\cline{3-7}
& \multirow{2}{*}{6} & $[2,2,1,1,1,1]$ & 1 & 1 & 15 & 33\\\cline{3-7}
&  & $[3,1,1,1,1,1]$ & 1 & 2 & 15 & 33\\\hline
\end{tabular}
\end{center}
\end{table}

\begin{itemize}
\item $\#_{\mathrm{NEW}}$ denotes the number of polynomials appearing in the
new condition (Theorem \ref{thm:multdiscriminant})

\item $\#_{\mathrm{YHZ}}$ denotes the number of polynomials appearing in the
YHZ's condition

\item $d_{\mathrm{NEW}}$ denotes the degree of the polynomial $D_{\mu}$
appearing in the new condition (Theorem \ref{thm:multdiscriminant})

\item $d_{\mathrm{YHZ}}$ denotes the maximum of the degrees of the polynomials
appearing in the YHZ's condition.
\end{itemize}

\noindent We make a few observations on the table.

\begin{enumerate}
\item Concerning the number of polynomials:

\begin{enumerate}
\item Observe that $\#_{\mathrm{NEW}}=1$ always. It is obvious from Theorem
\ref{thm:multdiscriminant}.

\item Observe that $\#_{\mathrm{YHZ}}=1$ when the entries of $\mu$ are at most
$2$ and that $\#_{\mathrm{YHZ}}$ is large when some entries of $\mu$ are
large. In fact, the observations hold in general, since straightforward
book-keeping of YHZ's algorithm immediately shows that%
\[
\#_{\mathrm{YHZ}}=1+\sum\limits_{i=1}^{m}{\binom{\mu_{i}-1}{2}}%
\]
For a proof, see Lemma \ref{lem:size_YHZ} in Appendix.

\item Hence $\#_{\mathrm{NEW}}\leq\ \#_{\mathrm{YHZ}}$ always and $=$ holds
only when the entries of $\mu$ are at most $2$.
\end{enumerate}

\item Concerning the maximum degree of polynomials:

\begin{enumerate}
\item Observe that $d_{\mathrm{NEW}}\leq2n-1=15$.
\revise{Recall that
\[
d_{\mathrm{NEW}}=2\,n-\mu_{m}\]
in Remark \ref{rem:degree_of_D}}.

\item Observe that $d_{\mathrm{YHZ}}\geq2n-1=15$ and that $d_{\mathrm{YHZ}}$
is large when some entries of $\mu$ are large. In fact, the observations is
conjectured to hold in general, since it can be shown, under some minor and
reasonable assumption, that
\[
d_{\mathrm{YHZ}}\geq2n+3^{\mu_{2}}-4\mu_{2}%
\]
For a proof, see Lemma \ref{lem:size_YHZ} in Appendix.

\item Hence most likely $\#_{\mathrm{NEW}}\leq\ \#_{\mathrm{YHZ}}$ always.
\end{enumerate}
\end{enumerate}

\medskip\noindent\textbf{Acknowledgements.} The second author's work was
supported by National Natural Science Foundation of China (Grant No. 11801101)
and Guangxi Science and Technology Program (Grant No. 2017AD23056).

\appendix

\section*{Appendix: Analysis of size of YHZ's condition}

We reproduce the result for the complex root case of the YHZ's method
\label{lem:YHZ} for readers' convenience. Assume $F$ is of degree $n$ with $m$
distinct roots. Let $\mu$ be an $m$-partition of $n$. Then we have%
\[
\operatorname*{mult}\left(  F\right)  =\mu\ \ \ \ \Longleftrightarrow
\ \ \ \left(  \bigwedge\limits_{i=1}^{\mu_{1}-2}\ \ \bigwedge\limits_{j=0}%
^{s_{i+1}-1}\overline{S_{j}\left(  G_{i}\right)  }=0\right)  \ \ \ \wedge
\ \ \overline{S_{0}\left(  G_{\mu_{1}-1}\right)  }\neq0
\]
where

\begin{itemize}
\item $S_{k}(G)=$ the $k$-th subresultant of $G\ $and $G^{\prime}$.

\item $\overline{S_{k}}=$ the coefficient of $x^{k}$ in $S_{k}\left(
G\right)  $

\item $s_{i}=\sum_{j=1}^{m}\max(\mu_{j}-i,0)$

\item $G_{i}=\left\{
\begin{array}
[c]{lll}%
F & \text{if } & i=0\\
S_{s_{i}}\left(  G_{i-1}\right)  & \text{if} & i>0
\end{array}
\right.  $
\end{itemize}

\begin{assumption}
\label{assumption:nonzero} We will assume that $\overline{S_{0}\left(
G_{j}\right)  }$ in the above YHZ's condition is not identically $0$ as a polynomial
on the coefficients of $F$.
\end{assumption}

\noindent We make this assumption because

\begin{enumerate}
\item It simplifies the analysis of the size of the condition produced by the
YHZ's method.

\item Numerous direct computations support its truth.

\item However, so far, we were not able to prove it.
\end{enumerate}

\begin{lemma}
[Size of YHZ's condition]\label{lem:size_YHZ}Under Assumption
\ref{assumption:nonzero}, we have

\begin{enumerate}
\item $\#_{\mathrm{YHZ}}=1+\sum\limits_{i=1}^{m}{\binom{\mu_{i}-1}{2}}$.

\item $d_{\mathrm{YHZ}}=\prod\limits_{j=0}^{\mu_{2}-1}(2\,m_{j}-1)\left\{
\begin{array}
[c]{lll}%
1 & \text{if} & \mu_{1}=\mu_{2}\\
1+\frac{2}{2m_{\mu_{2}-1}-1} & \text{if} & \mu_{1}=\mu_{2}+1\\
\left(  2\left(  \mu_{1}-\mu_{2}\right)  -1\right)  & \text{if} & \mu_{1}%
>\mu_{2}+1
\end{array}
\right.  $

where $m_{i}$ is the largest $k$ such that $\mu_{k}>i$.

\item $d_{\mathrm{YHZ}}\geq2n+3^{\mu_{2}}-4\mu_{2}$.
\end{enumerate}
\end{lemma}

\begin{proof}
We will prove each one by one.

\begin{enumerate}
\item Immediate from%
\[
\#_{\mathrm{YHZ}}=1+\sum_{j=2}^{\mu_{1}-1}s_{j}=1+\sum_{j=2}^{\mu_{1}-1}%
\sum_{i=1}^{m}\max(\mu_{i}-j,0)=1+\sum_{i=1}^{m}\sum_{j=2}^{\mu_{1}-1}\max
(\mu_{i}-j,0)=1+\sum_{i=1}^{m}\sum_{j=2}^{\mu_{i}-1}(\mu_{i}-j)=1+\sum
_{i=1}^{m}{\binom{\mu_{i}-1}{2}}.
\]

\item The proof is a bit long and so we divide it into several steps.

\begin{enumerate}
\item Note%
\begin{align*}
d_{\mathrm{YHZ}} &  =\max\left(  \left(  \bigcup\limits_{i=1}^{\mu_{1}%
-2}\left\{  \deg\ \overline{S_{0}\left(  G_{i}\right)  },\ldots,\deg
\overline{S_{s_{i+1}-1}\left(  G_{i}\right)  }\right\}  \right)
\bigcup\left\{  \deg\overline{S_{0}\left(  G_{\mu_{1}-1}\right)  }\right\}
\right)  \\
&  =\max_{1\leq i\leq\mu_{1}-1}\deg\ \overline{S_{0}\left(  G_{i}\right)  }\\
&  =\max_{1\leq i\leq\mu_{1}-1}\deg_{a}(G_{i})\left(  2s_{i}-1\right)  \\
&  =\max_{1\leq i\leq\mu_{1}-1}\left(  \prod_{j=0}^{i-1}(2\,m_{j}-1)\right)
\left(  2s_{i}-1\right)  \ \ \text{since the size of matrices for the
coefficients of }G_{j+1}\ \ \text{is \ }2m_{j}-1\\
&  =\max_{1\leq i\leq\mu_{1}-1}\left(  \prod_{j=0}^{i-1}(2\,m_{j}-1)\right)
\left(  2\sum_{j=i}^{\mu_{1}-1}m_{j}-1\right)
\end{align*}

\item The above motivates the following notations.
\begin{align*}
d_{i}  &  =\left(  \prod_{j=0}^{i-1}(2\,m_{j}-1)\right)  \left(  2\sum
_{j=i}^{\mu_{1}-1}m_{j}-1\right) \\
A  &  =\prod_{j=0}^{i-2}(2\,m_{j}-1)\ \ \,\text{and \ }\ B=\sum_{j=i}^{\mu
_{1}-1}m_{j}\ \ \text{and }C=m_{i-1}%
\end{align*}

\item We need to find $i$ such that $d_{i}$ is the maximum. Note
\begin{align*}
d_{i}\leq d_{i-1}\ \  &  \Longleftrightarrow\ \ A\left(  2C-1\right)  \left(
2B-1\right)  -A\left(  2\left(  B+C\right)  -1\right)  \leq0\\
&  \Longleftrightarrow\ \ 4A\left(  \left(  B-1\right)  \left(  C-1\right)
-\frac{1}{2}\right)  \leq0\\
&  \Longleftrightarrow\ \ B=1\ \ \vee\ \ C=1\ \ \ \ \ \text{since }%
A,B,C\geq1\\
&  \Longleftrightarrow\ \ \sum_{j=i}^{\mu_{1}-1}m_{j}=1\ \ \vee\ \ m_{i-1}=1\\
&  \Longleftrightarrow\ \ \ \left(  m_{i}=1\ \wedge\ i=\mu_{1}-1\right)
\ \ \vee\ \ m_{i-1}=1\\
&  \Longleftrightarrow\ \ \ \mu_{1}>\mu_{2}\ \ \wedge\ \ \left(  i=\mu
_{1}-1\ \ \vee\ \ i>\mu_{2}\right) \\
&  \Longleftrightarrow\ \ \ \left(  \mu_{1}=\mu_{2}+1\ \ \wedge\ \ \left(
i=\mu_{1}-1\ \ \vee\ \ i>\mu_{2}\right)  \right)  \ \ \vee\ \ \left(  \mu
_{1}>\mu_{2}+1\ \ \wedge\ \ \left(  i=\mu_{1}-1\ \ \vee\ \ i>\mu_{2}\right)
\right) \\
&  \Longleftrightarrow\ \ \ \left(  \mu_{1}=\mu_{2}+1\ \ \wedge\ \ i\geq
\mu_{2}\right)  \ \ \vee\ \ \left(  \mu_{1}>\mu_{2}+1\ \ \wedge\ \ i>\mu
_{2}\right)
\end{align*}

\item Thus%
\begin{align*}
d_{\mathrm{YHZ}} &  =\max_{1\leq i\leq\mu_{1}-1}d_{i}\\
&  =\left\{
\begin{array}
[c]{lll}%
d_{\mu_{2}-1} & \text{if} & \mu_{1}=\mu_{2}\\
d_{\mu_{2}-1} & \text{if} & \mu_{1}=\mu_{2}+1\\
d_{\mu_{2}} & \text{if} & \mu_{1}>\mu_{2}+1
\end{array}
\right.  \\
&  =\left\{
\begin{array}
[c]{lll}%
\left(  \prod_{j=0}^{\mu_{2}-1-1}(2\,m_{j}-1)\right)  \left(  2\sum_{j=\mu
_{2}-1}^{\mu_{1}-1}m_{j}-1\right)   & \text{if} & \mu_{1}=\mu_{2}\\
\left(  \prod_{j=0}^{\mu_{2}-1-1}(2\,m_{j}-1)\right)  \left(  2\sum_{j=\mu
_{2}-1}^{\mu_{1}-1}m_{j}-1\right)   & \text{if} & \mu_{1}=\mu_{2}+1\\
\left(  \prod_{j=0}^{\mu_{2}-1}(2\,m_{j}-1)\right)  \left(  2\sum_{j=\mu_{2}%
}^{\mu_{1}-1}m_{j}-1\right)   & \text{if} & \mu_{1}>\mu_{2}+1
\end{array}
\right.  \\
&  =\left\{
\begin{array}
[c]{lll}%
\left(  \prod_{j=0}^{\mu_{2}-2}(2\,m_{j}-1)\right)  \left(  2m_{\mu_{2}%
-1}-1\right)   & \text{if} & \mu_{1}=\mu_{2}\\
\left(  \prod_{j=0}^{\mu_{2}-2}(2\,m_{j}-1)\right)  \left(  2\left(
1+m_{\mu_{2}-1}\right)  -1\right)   & \text{if} & \mu_{1}=\mu_{2}+1\\
\left(  \prod_{j=0}^{\mu_{2}-1}(2\,m_{j}-1)\right)  \left(  2\left(  \mu
_{1}-\mu_{2}\right)  -1\right)   & \text{if} & \mu_{1}>\mu_{2}+1
\end{array}
\right.  \\
&  =\left\{
\begin{array}
[c]{lll}%
\left(  \prod_{j=0}^{\mu_{2}-1}(2\,m_{j}-1)\right)   & \text{if} & \mu_{1}%
=\mu_{2}\\
\left(  \prod_{j=0}^{\mu_{2}-2}(2\,m_{j}-1)\right)  \left(  2m_{\mu_{2}%
-1}+1\right)   & \text{if} & \mu_{1}=\mu_{2}+1\\
\left(  \prod_{j=0}^{\mu_{2}-1}(2\,m_{j}-1)\right)  \left(  2\left(  \mu
_{1}-\mu_{2}\right)  -1\right)   & \text{if} & \mu_{1}>\mu_{2}+1
\end{array}
\right.  \\
&  =\prod\limits_{j=0}^{\mu_{2}-1}(2\,m_{j}-1)\left\{
\begin{array}
[c]{lll}%
1 & \text{if} & \mu_{1}=\mu_{2}\\
1+\frac{2}{2m_{\mu_{2}-1}-1} & \text{if} & \mu_{1}=\mu_{2}+1\\
\left(  2\left(  \mu_{1}-\mu_{2}\right)  -1\right)   & \text{if} & \mu_{1}%
>\mu_{2}+1
\end{array}
\right.
\end{align*}

\end{enumerate}

\item We will divide the proof into three cases. \newpage

\begin{enumerate}
\item $\mu_{1}=\mu_{2}$.

\begin{enumerate}
\item We rewrite%
\[
d_{\mathrm{YHZ}}=\prod_{j=0}^{\mu_{2}-1}(2\,m_{j}-1)=\left(  \prod_{j=0}%
^{\mu_{2}-1}(2\,m_{j}-1)\right)  -\left(  \sum_{j=0}^{\mu_{2}-1}2m_{j}\right)
+2n.
\]

\item Let
\[
G(x_{0},\ldots,x_{\mu_{2}-1})=\prod_{j=0}^{\mu_{2}-1}(x_{j}-1)-\sum_{j=0}%
^{\mu_{2}-1}x_{j}+2n
\]
over
\[
R=\{(x_{0},\ldots,x_{\mu_{2}-1}):\,x_{j}\geq2\cdot2\}.
\]

Then we have%
\[
d_{\mathrm{YHZ}}\geq\min_{x\in R}G\left(  x\right)
\]

\item Note
\[
\partial G/\partial x_{i}=\prod_{\substack{0\leq j\leq\mu_{2}-1\\j\neq
i}}(x_{j}-1)-1>0\ \text{over }R
\]
Hence%
\[
\min_{x\in R}G\left(  x\right)  =\prod_{j=0}^{\mu_{2}-1}(4-1)-\sum_{j=0}%
^{\mu_{2}-1}4+2n=3^{\mu_{2}}-4\mu_{2}+2n
\]

\item Hence%
\[
d_{\mathrm{YHZ}}\geq3^{\mu_{2}}-4\mu_{2}+2n
\]

\end{enumerate}

\item $\mu_{1}=\mu_{2}+1$.

\begin{enumerate}
\item We rewrite%
\begin{align*}
d_{\mathrm{YHZ}} &  =\left(  \prod_{j=0}^{\mu_{2}-1}(2\,m_{j}-1)\right)
\cdot\left(  1+\frac{2}{2m_{\mu_{2}-1}-1}\right)  \\
&  =\left(  \prod_{j=0}^{\mu_{2}-2}(2\,m_{j}-1)\right)  \left(  2m_{\mu_{2}%
-1}+1\right)  -\left(  2+\sum_{j=0}^{\mu_{2}-1}2m_{j}\right)  +2n.
\end{align*}

\item Let
\[
G(x_{0},\ldots,x_{\mu_{2}-1})=(x_{\mu_{2}-1}+1)\prod_{j=0}^{\mu_{2}-2}%
(x_{j}-1)-\left(  2+\sum_{j=0}^{\mu_{2}-1}x_{j}\right)  +2n
\]
over
\[
R=\{(x_{0},\ldots,x_{\mu_{2}-1}):\,x_{j}\geq2\cdot2\}.
\]

Then we have%
\[
d_{\mathrm{YHZ}}\geq\min_{x\in R}G\left(  x\right)
\]

\item Note
\[
\partial G/\partial x_{i}=\left\{
\begin{array}
[c]{ll}%
(x_{\mu_{2}-1}+1)\prod_{\substack{0\leq j\leq\mu_{2}-2\\j\neq i}%
}(x_{j}-1)-1, & \mathrm{if}\ \ i<\mu_{2}-1\\
\prod_{0\leq j\leq\mu_{2}-2}(x_{j}-1)-1, & \mathrm{if}\ \ i=\mu_{2}-1
\end{array}
\right.
\]
and $\partial G/\partial x_{i}>0$ over $R$. Hence%
\begin{align*}
\min_{x\in R}G\left(  x\right)   &  =(4+1)\prod_{j=0}^{\mu_{2}-2}(4-1)-\left(
2+\sum_{j=0}^{\mu_{2}-1}4\right)  +2n\\
&  =5\cdot3^{\mu_{2}-1}-4\mu_{2}-2+2n\\
&  =\frac{5}{3}\cdot3^{\mu_{2}}-4\mu_{2}-2+2n\\
&  =\left(  3^{\mu_{2}}-4\mu_{2}\right)  +\left(  \frac{2}{3}\cdot3^{\mu_{2}%
}-2\right)  +2n\\
&  \geq3^{\mu_{2}}-4\mu_{2}+2n
\end{align*}

\item Hence%
\[
d_{\mathrm{YHZ}}\geq3^{\mu_{2}}-4\mu_{2}+2n
\]

\end{enumerate}

\item $\mu_{1}>\mu_{2}+1$.

\begin{enumerate}
\item We rewrite%
\begin{align*}
d_{\mathrm{YHZ}} &  =\left(  \prod_{j=0}^{\mu_{2}-1}(2\,m_{j}-1)\right)
\left(  2\left(  \mu_{1}-\mu_{2}\right)  -1\right)  \\
&  =\left(  \prod_{j=0}^{\mu_{2}-1}(2\,m_{j}-1)\right)  \left(  2\left(
\mu_{1}-\mu_{2}\right)  -1\right)  -\left(  \sum_{j=0}^{\mu_{2}-1}2m_{j}%
+2(\mu_{1}-\mu_{2})\right)  +2n
\end{align*}

\item Let
\[
G(x_{0},\ldots,x_{\mu_{2}})=\prod_{j=0}^{\mu_{2}}(x_{j}-1)-\sum_{j=0}^{\mu
_{2}}x_{j}+2n
\]
over
\[
R=\{(x_{0},\ldots,x_{\mu_{2}}):\,x_{j}\geq2\cdot2\}.
\]

Then we have%
\[
d_{\mathrm{YHZ}}\geq\min_{x\in R}G\left(  x\right)
\]

\item Note
\[
\partial G/\partial x_{i}=\prod_{\substack{0\leq j\leq\mu_{2}\\j\neq i}%
}(x_{j}-1)-1>0\ \text{over }R
\]
Hence%
\begin{align*}
\min_{x\in R}G\left(  x\right)   &  =\prod_{j=0}^{\mu_{2}}(4-1)-\sum
_{j=0}^{\mu_{2}}4+2n\\
&  =3^{\mu_{2}+1}-4\left(  \mu_{2}+1\right)  +2n\\
&  =3\cdot3^{\mu_{2}}-4\left(  \mu_{2}+1\right)  +2n\\
&  =\left(  3^{\mu_{2}}-4\mu_{2}\right)  +\left(  2\cdot3^{\mu_{2}}-4\right)
+2n\\
&  \geq3^{\mu_{2}}-4\mu_{2}+2n
\end{align*}

\item Hence%
\[
d_{\mathrm{YHZ}}\geq3^{\mu_{2}}-4\mu_{2}+2n
\]

\end{enumerate}
\end{enumerate}
\end{enumerate}
\end{proof}

\end{document}